\documentclass[letterpaper, 11pt]{amsart}
\usepackage{macros}
\usepackage{fullpage}
\usepackage{notations}

\usepackage[foot]{amsaddr}
\newboolean{doubleblind}
\setboolean{doubleblind}{False}
\title{An FPRAS for Antiferromagnetic Ising Models on Random Regular Bipartite Graphs}
\ifdoubleblind
    \author{Author(s)}
\else
    \author{Zhidan Li, Kuan Yang}
    \address[Zhidan Li]{School of Computer Science, Shanghai Jiao Tong University, Shanghai, China. \textnormal{Email: \url{yueEnTeRnAL@sjtu.edu.cn}}.}
    \address[Kuan Yang]{John Hopcroft Center for Computer Science, Shanghai Jiao Tong University, Shanghai, China. \textnormal{Email: \url{kuan.yang@sjtu.edu.cn}}.}
\fi

\begin{document}

\begin{abstract}
We design randomized approximation schemes for the partition function of antiferromagnetic Ising models with uniform external field on random regular bipartite graphs.
Our algorithm generalizes the approach of Kocurek, Oveis Gharan and Tjowasi (arXiv, 2026) for hard-core models on the same random graph model beyond the uniqueness threshold.
We show that, as long as $\lambda$ is upper bounded by a constant and $\lambda(1 - \beta) \lesssim \Delta^{-1/2}$, an efficient randomized algorithm approximates the partition function with high probability.
The algorithm first truncates configurations that are large on either side of the bipartition and then samples from Gibbs distributions conditioned on fixed sizes on one or both sides.
To choose an optimal truncation bound, we establish concentration properties of the Gibbs distribution on random regular bipartite graphs. Then we apply high-dimensional expansion and prove trickle-down theorems to obtain fast samplers for the conditioned distributions.
\end{abstract}

\maketitle

\section{Introduction} \label{sec:introduction}
The \emph{(Lenz-)Ising model} is one of the most fundamental models in statistical physics, serving as a canonical framework for studying phase transitions, Gibbs measures, and spin correlations.
From the algorithmic perspective, approximating the partition function of the Ising model is a central problem in approximate counting, with applications ranging from statistical inference to sampling from Gibbs distributions.
%Approximating its partition function is a fundamental algorithmic problem in statistical physics, and also a canonical test case for understanding the boundary between tractability and computational hardness of statistical physics systems.
We study the Ising model in its subgraph representation, in which configurations are subsets of vertices. This representation is equivalent to the usual spin formulation on regular graphs, and is particularly convenient for bipartite graphs because the configuration can be naturally decomposed according to the two sides of the bipartition. Given a graph $G = (V, E)$ and two parameters $\lambda > 0$ and $\beta \ge 0$, the Ising model with uniform external field $\lambda$ and edge interaction $\beta$ on $G$ consists of a finite state space (or configuration space) $\Omega = 2^{V(G)}$ and a weight function $w = w_{G; \lambda, \beta} : \Omega \to \mathbb{R}$ inducing a partition function $\Z_G(\lambda, \beta)$ and its Gibbs distribution $\mu = \mu_{G; \lambda, \beta}$.
A configuration $S$ in $\Omega$ is a subset of vertices $V$, and its weight is defined by
\begin{align} \label{eq:weight-function}
    w(S) = \lambda^{\abs{S}} \beta^{\abs{E(G[S])}}
\end{align}
where $G[S]$ is the induced subgraph of $G$ by $S$.
The \emph{partition function} of the model is the total weight of the system:
\begin{align*}
    \Z_{G}(\lambda, \beta) = \sum_{S \subseteq V(G)} w(S),
\end{align*}
which is the normalizing factor of the induced Gibbs distribution $\mu_{G; \lambda, \beta}$ on $\Omega$:
\begin{align*}
    \forall S \subseteq V(G), \quad \mu(S) = \frac{w(S)}{\Z_{G}(\lambda, \beta)}.
\end{align*}
When $\beta > 1$, the model is \emph{ferromagnetic}; when $\beta < 1$, it is \emph{antiferromagnetic}.
When the underlying graph $G$ is regular, this subgraph representation is equivalent, up to a change of parameters, to the usual spin-system formulation; see~\Cref{sec:model-equivalence}.

In this work, we focus on antiferromagnetic Ising models.
It is a significant problem to estimate their partition functions, because the partition function determines the free energy and  approximate counting is closely related to sampling from Gibbs distributions and estimating thermodynamic quantities.
Unfortunately, exact computation is $\numberP$-hard even on triangle-free graphs, since counting independent sets ($\#\mathbf{IS}$), which can be viewed as a special case  ($\lambda = 1$ and $\beta = 0$) of anti-ferromagnetic Ising models, is $\numberP$-hard ~\cite{Greenhill00}.
This motivates the search for efficient approximation schemes for $\Z_{G}(\lambda, \beta)$.

In the study of approximate counting, the uniqueness threshold of the Gibbs measure on infinite $\Delta$-regular trees has been shown to characterize the algorithmic threshold for many important classes of spin systems.
On the algorithmic side, seminal works ~\cite{Weitz06,LLY13,CCYZ25,CJMYZ26} show that, in the uniqueness regime or at the threshold, there exists a (randomized) polynomial-time algorithm to approximate $\Z_G(\lambda, \beta)$.
In contrast, Sly and Sun~\cite{SS12}, and Galanis, \v{S}tefankovi\v{c} and Vigoda~\cite{GSV16} show that unless $\NP = \RP$, there is no $\FPRAS$ for approximating the partition functions for general bounded-degree graphs in the non-uniqueness regime. This dichotomy suggests that the tree uniqueness threshold captures an important boundary in the algorithmic complexity of spin systems.
However, for bipartite graphs, the situation is different. The problem of (approximately) counting independent sets on a bipartite graph ($\numberBIS$) is widely regarded as an intermediate approximate counting problem between $\P$ and $\NP$; currently no computational hardness results are known, nor is there an approximation algorithm that works for general cases. In addition, Cai et al.~\cite{CGGGJSV16} shows that approximating the partition function of an antiferromagnetic two-spin system in the non-uniqueness regime is $\numberBIS$-hard.

Despite the lack of approximation schemes on general graphs in the tree non-uniqueness threshold, recent studies has shown that random bipartite graph models provide an important intermediate setting between trees and worst-case graphs. They retain nontrivial global dependence structure responsible for phase transitions, while their expansion properties allow additional algorithmic techniques. Understanding spin systems on such graphs may reveal whether the tree uniqueness threshold remains an algorithmic barrier in typical instances. For example, approximation schemes beyond the tree uniqueness threshold have recently been obtained for the hard-core model on random regular bipartite graphs: 
Kocurek, Oveis Gharan and Tjowasi~\cite{KOGT26} give an $\FPRAS$ whenever $\lambda = O(\Delta^{-1/2})$. Together with~\cite{JKP20}, these works give approximation schemes for the hard-core model for all fugacities on random regular bipartite graphs.
Meanwhile, for antiferromagnetic Ising models, Geisler et al.~\cite{GKSW26} show that, despite torpid mixing of Glauber dynamics, an $\FPTAS$ exists for $\Z_G(\lambda, \beta)$ on a family of regular bipartite expanders that includes almost every regular bipartite graph, in a regime far from the uniqueness threshold.
These algorithmic results motivate the following question:
\begin{center}
\emph{Does the antiferromagnetic Ising model on a random regular bipartite graph admit an efficient deterministic or randomized approximation scheme for all uniform external fields and all edge interactions, with high probability?}
\end{center}

We partially resolve this question by establishing an $\FPRAS$ in a weak-interaction regime that extends beyond the tree uniqueness threshold. Our main result is the following theorem.
\begin{theorem} \label{thm:random-regular-bipartite-Ising-FPRAS}
    There exist a positive integer $\Delta_0 \ge 3$ and two positive constants $\lambda_0, C > 0$ such that for every positive integer $\Delta \ge \Delta_0$, real numbers $\lambda > 0$ and $\beta \in [0, 1]$ such that
    \begin{align*}
        \lambdaUpperBound, \quad \ParameterCondition,
    \end{align*}
    the following holds with high probability\footnote{In this paper, when we say that an event occurs `with high probability', we mean that it occurs with probability \(1-o(1)\).} over the random $\Delta$-regular bipartite $2n$-vertex graph $G$ for all sufficiently large $n$.
    There is an $\FPRAS$ for the partition function of Ising models on $G$ with uniform external field $\lambda$ at edge interaction $\beta$.
\end{theorem}
We emphasize that the boundedness assumption on $\lambda$ is essential in our current analysis, since it plays a crucial role to establish the concentration properties of the configuration space. Thus, the theorem concerns a weak-interaction regime rather than the full parameter space.

\subsection{Overview of techniques}
The main difficulty is that Glauber dynamics may mix slowly beyond the uniqueness threshold. Rather than analyzing local dynamics on the full configuration space directly, we exploit the geometry of suitable conditioned configuration spaces, decompose the Gibbs distribution into a collection of structured slices and prove rapid mixing within each slice; details appear in~\Cref{sec:preliminaries,sec:random-regular-bipartite-Ising}.

\subsubsection*{\bf Concentration phenomenon}
A central geometric property of $\Omega$ is \emph{concentration}.
We establish two kinds of concentrations: (a) the one-side concentration inequality, and (b) the concentration for \emph{degree profiles}.
The one-side concentration inequality gives exponential decay for the number of selected vertices on either side, and the other concentration provides a demonstration of the degree profiles after fixing a vertex subset on one side of size at most half of that side.
Together, these estimates allow us to truncate large configurations and reduce the problem to small or balanced `slices' on both sides, while retaining an accurate approximation of $\Z_G(\lambda,\beta)$.

\subsubsection*{\bf High-dimensional expanders}
High-dimensional expanders provide an efficient framework to design and analyze $\FPAUS$es for counting problems.
Seminal works~\cite{AL20,Oppenheim18} develop tools for deriving rapid mixing of \emph{down-up walks} from `local' properties of the Gibbs distribution.

We apply this framework to analyze the down-up walks for the Gibbs distribution conditional on specific slices.
To establish local spectral expansion, we prove trickle-down theorems analogous to those in~\cite{KOGT26}, with modifications required for the soft constraints of the antiferromagnetic Ising model. However, extending these techniques is nontrivial. Unlike the hard-core model, which imposes a hard constraint forbidding adjacent occupied vertices, the antiferromagnetic Ising model allows configurations containing adjacent occupied vertices but receiving smaller weights. Consequently, the correlations inside the Gibbs distribution are more subtle, and the techniques developed for hard constraints cannot be directly applied.

\subsection{Related works and discussion}

Our result gives an $\FPRAS$ when both $\lambda$ and $\lambda(1-\beta)\sqrt{\Delta}$ are upper bounded by constants; this regime includes a portion of the tree non-uniqueness region.
In contrast, ~\cite{GKSW26} provides an $\FPTAS$ in a different non-uniqueness regime, when $\lambda$ is upper bounded and $\lambda(1 - \beta) = \Omega\ab((\ln{\Delta})^{3/2}/\Delta^{\min\set{1, 1 - \kappa/2}})$, via abstract polymer models and cluster expansion approach.
A gap remains between these parameter regimes.
It also remains open whether  the upper bound on the vertex fugacity $\lambda$ is inherent or merely a limitation of the current methods.

To analyze the configuration space, our analysis adapts the \emph{independent-set slices} of~\cite{KOGT26} to this soft-constrained setting.
The concentration estimates identify the truncation parameters, while high-dimensional expanders are used to prove rapid mixing of down-up walks on slices of prescribed sizes.
Similar local-to-global arguments have been used to sample fixed-size independent sets~\cite{AL20,JMPV23}, sample hard-core and antiferromagnetic two-spin systems on random regular graphs~\cite{CCCYZ25}, and approximately count the hard-core models on random regular bipartite graphs~\cite{KOGT26}.
%We present trickle-down theorems like~\cite{KOGT26} and derive an $\FPRAS$ by the rapid sampler for slices with desired size.

\subsection{Organization of this paper}
We introduce necessary preliminaries in~\Cref{sec:preliminaries}.
In~\Cref{sec:random-regular-bipartite-Ising}, we establish some properties of the Gibbs distribution of antiferromagnetic Ising models on random regular bipartite graphs.
Proofs and analysis of trickle-down theorems of high-dimensional expanders are provided in~\Cref{sec:trickle-theorems}.
Finally, we develop our approximation scheme in~\Cref{sec:algorithm-design}.

\subsection*{Statement on AI use}
We apply ChatGPT 5.6 Pro Sol to optimize the choice of parameters in our algorithm.
The concentration inequalities are provided by ChatGPT 5.6 Pro Sol after we give our drafts and related materials to it.
All calculations in this work are verified by us and we take full responsibility of this paper.
\section{Preliminaries} \label{sec:preliminaries}
\subsection{Notations}
We use $\e$ to denote the natural logarithm base and $\ln$ to denote the natural logarithm function $\log_{\e}$.
For a natural number $n$, we use $[n]$ to denote the set $\set{1, 2, \ldots, n}$.

For an operator/matrix $A$, we use $A^{\top}$ to denote its transpose.
The spectral norm of $A$ is defined by $\norm{A}_2 \defeq \max_{\norm{\vecx}_2 = \norm{\vecy}_2 = 1}\vecx^{\top} A \vecy$.
We use $\one_{S}$ to denote the all-one vector on $S$ and $\ID_S$ to denote the identity matrix on $S$ (when context is clear, we simply use $\one$ and $\ID$).
We use $\lambda_i(A)$ to denote the $i$-th largest eigenvalue of $A$ and $\lambda_{\max}(A), \lambda_{\min}(A)$ to denote the largest and smallest eigenvalue respectively.

For a random variable $X$, we denote by $\mathsf{Law}(X)$ its law.
The expectation of $X$ is denoted by $\E{X}$ and the variance of $X$ is denoted by $\Var{X}$.
We use $M_{X}(\cdot)$ to denote the moment generating function of $X$, \IE, $M_X(t) = \E{e^{tX}}$.
The entropy function $H(X)$ is defined as $H(X) \defeq -\sum_{x} \Pr{X = x} \ln{\Pr{X = x}}$ with convention $0\ln{0} = 0$.
With abuse of notations, we use $H(p)$ to denote the relative entropy function of a Bernoulli random variable with rate $p$.

For two probability distributions $P$ and $Q$ supported a finite state space $\Omega$, the Kullback-Leibler divergence $\KL{P}{Q}$ is defined by
\begin{align*}
    \KL{P}{Q} \defeq \sum_{x \in \Omega} P(x) \ln{\frac{P(x)}{Q(x)}},
\end{align*}
with convention $0\ln{(0/0)} = 0$ and $p\ln{(p/0)} = \infty$ for $p > 0$.
For two real numbers $p, q \in [0, 1]$, we simply use $\KL{p}{q}$ to denote the KL-divergence between two Bernoulli random variables with mean $p$ and $q$, \IE, $\KL{p}{q} \defeq \KL{\Ber{p}}{\Ber{q}}$.

The graph $G = (V(G), E(G))$ consists of a vertex set $V(G)$ and an edge set $E(G) \subseteq V(G) \times V(G)$.
The bipartite graph is denoted by $G = (V(G) = V_L(G) \cup V_R(G), E(G))$ where $V_L(G) \cap V_R(G) = \emptyset$ and $E(G) \subseteq V_L(G) \times V_R(G)$.
When the context is clear, we simply use $V, V_L, V_R$ and $E$ to denote the graph.
Given a graph $G = (V, E)$, let $E_G(S, T)$ be the number of edges between $S$ and $T$.
For a vertex subset $S \subseteq V$, let $G[S]$ be the subgraph of $G$ induced by $S$.
Moreover, we define the degree function $d_S(\cdot)$ on $V$ as $d_S(v) \defeq \abs{E_G(\set{v}, S)}$ the number of edges between $v$ and $S$.

\subsection{Useful inequalities}
We state some useful inequalities here.
\begin{lemma}[Markov inequality] \label{lem:Markov-inequality}
    For a non-negative random variable $X$ and $a > 0$, it holds that
    \begin{align*}
        \Pr{X \ge a} \le \frac{\E{X}}{a}.
    \end{align*}
    Moreover, if it always holds that $X \ge 1$, then $\Pr{X > 0} \le \E{X}$.
\end{lemma}

Using Markov inequality with the moment generating function, we can derive the following version of Chernoff bound (for details, refer~\cite{Gerbessiotis25}).
\begin{lemma}[Chernoff bound] \label{lem:Chernoff-bound}
    For a non-negative random variable $X$ and $a, t > 0$, it holds that
    \begin{align*}
        \Pr{X \ge a} \le M_X(t) \e^{-ta}.
    \end{align*}
    Moreover, if $X = \sum_{i = 1}^n X_i$ where each $X_i$ is an independent Bernoulli random variable with rate $p$, for every $p < \rho < 1$, it holds that
    \begin{align*}
        \Pr{X \ge \rho n} \le \exp(-n\KL{\rho}{p}) = \ab(\ab(\frac{p}{\rho})^{\rho} \ab(\frac{1 - p}{1 - \rho})^{1 - \rho})^n.
    \end{align*}
\end{lemma}

The Popoviciu's inequality on variances in~\cite{Popoviciu35} gives an upper bound on the variance with respect to the range of the random variable.
\begin{proposition}[Popoviciu's inequality on variances] \label{prop:Popoviciu-inequality}
    For a random variable $X$, it always holds that
    \[
        \Var{X} \le \frac{1}{4} \ab(M - m)^2 
    \]
    where $M$ is the possible maximum of $X$ and $m$ is the possible minimum of $X$.
\end{proposition}

The following concentration inequality for martingales with bounded increments, which is a special version of Bernstein's inequality and known as \emph{Freedman's inequality}, plays a crucial role in our analysis.
For details, we recommend readers to refer~\cite{Freedman75,DvZ01}.
\begin{lemma}[Freedman's inequality] \label{lem:Freedman-inequality}
    Let $\set{(Z_i, \calF_i)}_{i = 0}^n$ be a martingale with bounded increments ($\abs{Z_i - Z_{i - 1}} \le B$ almost surely) and set
    \begin{align*}
        V_n = \sum_{i = 1}^n \E{(Z_i - Z_{i - 1})^2 \mid \calF_{i - 1}}.
    \end{align*}
    Then for every $t, v > 0$, it holds that
    \begin{align*}
        \Pr{\abs{Z_n - Z_0} \ge t, V_n \le v} \le 2\exp\ab(-\frac{t^2}{2(v + Bt/3)}).
    \end{align*}
\end{lemma}

\subsection{Markov chains and mixing time}
We state here some basic definitions and lemmas for Markov chains and for more details, please see~\cite{LP17}.

For a state space $\Omega$, a Markov chain $M$ consists of a transition kernel $P$ on $\Omega$ with its stationary distribution $\pi$.
For brevity, we sometimes use the transition matrix $P$ to denote the Markov chain $M$.
In this work, we only consider \emph{ergodic} and \emph{reversible} Markov chains (see details in~\cite{LP17}).
%, and for an ergodic Markov chain $M$, we say it is reversible with respect to $\pi$ if the detailed balanced equation holds:
%\begin{align} \label{eq:detailed-balanced-equation}
%    \forall \sigma, \tau \in \Omega, \quad \pi(\sigma) P(\sigma, \tau) = \pi(\tau) P(\tau, \sigma).
%\end{align}

We use standard notation for total variation distance and mixing time.
For two probability distributions $\mu, \nu$ on $\Omega$, the \emph{total variation distance} is defined as
\begin{align*}
    \TV{\mu}{\nu} \defeq \frac{1}{2} \sum_{x \in \Omega} \abs{\mu(x) - \nu(x)} = \max_{\Lambda \subseteq \Omega}\ \abs{\mu(\Lambda) - \nu(\Lambda)}.
\end{align*}
For a positive real $\eps \in (0, 1)$, define the \emph{mixing time} of $M$ with error $\eps$ as
\begin{align*}
    \tau_{\mix}(M, \eps) \defeq \inf\set{t \in \mathbb{N} \midc \forall x \in \Omega, \TV{P^t(x, \cdot)}{\pi} \le \eps}.
\end{align*}
It is a route way to bound the mixing time by the spectral gap of $P$.
For a random walk $P$, set $\lambda_*(P) = \max\set{\lambda_2(P), \abs{\lambda_{\min}(P)}}$.
The following fact is standard.
\begin{fact}
    For a Markov chain $P$ with respect to $\mu$, it holds that
    \[
        \tau_{\mix}(P, \eps) \le O\ab(\frac{1}{1 - \lambda_*(P)}\ab(\ln{\frac{1}{\mu_{\min}}} + \ln{\eps^{-1}}))
    \]
    where $\mu_{\min} \defeq \min_{x \in \Omega : \mu(x) > 0} \mu(x)$.
    Consequently, for its lazy version $\frac{P + I}{2}$, it holds that
    \[
        \tau_{\mix}\ab(\frac{P + I}{2}, \eps) \le O\ab(\frac{1}{1 - \lambda_2(P)}\ab(\ln{\frac{1}{\mu_{\min}}} + \ln{\eps^{-1}})).
    \]
\end{fact}

\subsection{Approximate algorithms}
We formally define the concept of approximate algorithms.
For a quantity $Z \in \mathbb{R}$ and a tolerance $\eps \in (0, 1)$, we say $\wh{Z}$ is an $\eps$-approximation to $Z$ if $(1 - \eps) Z \le \wh{Z} \le (1 + \eps) Z$.
\begin{definition}[{$\FPTAS$ and $\FPRAS$}]
    For a function $Z$ from a family of instances to reals, a \emph{fully polynomial-time approximation scheme} ($\FPTAS$) is a deterministic algorithm such that, for every instance $\Phi$ and every tolerance $\eps \in (0, 1)$, it outputs an $\eps$-approximation to $Z(\Phi)$ in time $\poly{\abs{\Phi}, \eps^{-1}}$.
    Similarly, a \emph{fully polynomial-time randomized approximation scheme} ($\FPRAS$) is a randomized algorithm such that, for every instance $\Phi$ and every tolerance $\eps \in (0, 1)$, it outputs an $\eps$-approximation to $Z(\Phi)$ in time $\poly{\abs{\Phi}, \eps^{-1}}$ with probability at least $3/4$.
\end{definition}

\begin{definition}[$\FPAUS$]
    Fix an input $\Phi$ and an induced state space $\Omega = \Omega(\Phi)$.
    For a non-zero weight function $w : \Omega \to \mathbb{R}_{\ge 0}$, consider the distribution $\mu$ on $\Omega$ defined as $\mu(x) \propto w(x)$.
    A \emph{fully polynomial-time almost-uniform sampler} ($\FPAUS$) for $\mu$ outputs a sample $X$ from $\Omega$ randomly provided a tolerance $\eps \in (0, 1)$ such that $\TV{\mathsf{Law}(X)}{\mu} \le \eps$ and the running time is $\poly{\abs{\Phi}, \eps^{-1}}$.
\end{definition}

For a family of self-reducible problems,~\cite{JVV86} shows that an $\FPAUS$ can induce an $\FPRAS$.
For more details, refer~\cite{JVV86}.

\subsection{High-dimensional expander}
Given a finite ground set $U$, a \emph{$d$-dimensional pure simplicial complex} $(\SC, \pi)$ is a downward closed collection of subsets of $U$ with all maximal subsets having size $d$ and  equipped with a probability distribution $\pi$ over all maximal subsets.
We call a subset in $\SC$ a \emph{face} or \emph{facet}, and for every facet $S \in \SC$, its dimension $\dim{S}$ is defined as its cardinality $\abs{S}$ and the co-dimension $\codim{S}$ is defined by $\codim{S} = d - \dim{S}$.
For every integer $k = 0, 1, \ldots, d$, let $\SC(k)$ be the collection of facets of dimension $k$, \IE,
\begin{align*}
    \SC(k) \defeq \set{S \in \SC \mid \dim{S} = k}.
\end{align*}
For each $k = 0, 1, \ldots, d$, we induce a probability distribution $\pi_k$ according to the marginal distribution of $\pi$ on $\SC(k)$ as
\begin{align} \label{eq:simplicial-complex-weight-function}
    \forall S \in \SC(k), \quad \pi_k(S) = \frac{1}{\binom{d}{k}} \sum_{T \in \SC(d) : S \subseteq T} \pi(T).
\end{align}
Based on $\set{\pi_{k}}_{k = 0}^d$, we define random walks on $\SC$.
At first, we consider two kinds of operators $\pu_k \in \mathbb{R}^{\SC(k) \times \SC(k + 1)}$ for $k = 0, \ldots, d - 1$ and $\pd_k \in \mathbb{R}^{\SC(k) \times \SC(k - 1)}$ for $k = 1, \ldots, d$ defined as
\begin{align*}
    \pu_k(S, T) = \id{T \subseteq S} \cdot \frac{\pi_{k + 1}(T)}{(k + 1)\pi_k(S)}, \quad \pd_k(S, T) = \id{T \subseteq S} \frac{1}{k}.
\end{align*}
Then there are two natural random walks on the simplicial complexes $\SC$:
\begin{enumerate}
    \item the up-down walk $\pud_{k} = \pu_{k} \pd_{k + 1}$ for every $k = 0, 1, \ldots, d - 1$ and its non-lazy version $\pudnl_{k} = \frac{k + 2}{k + 1} \pud_{k} - \frac{1}{k + 1} \ID_{\SC(k)}$;
    \item the down-up walk $\pdu_{k} = \pd_{k} \pu_{k - 1}$ for every $k = 1, \ldots, d$ and its non-lazy version $\pdunl_{k} = \frac{k + 1}{k} \pdu_{k} - \frac{1}{k} \ID_{\SC(k)}$.
\end{enumerate}
It is direct to observe that these random walks are all reversible with respect to $\pi_k$.

An important sub-structure of a simplicial complex plays a crucial role is the \emph{links}.
For a facet $S \in \SC$, we define the link of $\SC$ on $S$ as
\begin{align*}
    \SC_S \defeq \set{T \mid T \cap S = \emptyset \land S \cup T \in \SC}.
\end{align*}
It is similar to define the distributions $\pi_{S, k}$ and random walks $\pu_{S, k}$, $\pd_{S, k}$, $\pud_{S, k}$, $\pdu_{S, k}$, $\pudnl_{S, k}$ and $\pdunl_{S, k}$ on $\SC_S$ as notations in $\SC$.

The concept of \emph{local-spectral expander} plays a crucial role in the analysis of mixing rate of random walks on a simplicial complex.
\begin{definition}[Local-spectral expander] \label{def:local-spectral-expander}
    For a $d$-dimensional simplicial complex $(\SC, \pi)$, we say it is a \emph{$(\gamma_0, \ldots, \gamma_{d - 2})$-local-spectral expander} if for every $k = 0, \ldots, d - 2$ and $S \in \SC(k)$, it holds that $\lambda_{2}(\pudnl_{S, 1}) \le \gamma_k$.
\end{definition}

When $\SC$ is a local-spectral expander, the spectral gap of the down-up walk $\pud_{d}$ is bounded by~\cite{AL20}.
\begin{theorem}[\cite{AL20}] \label{thm:local-to-global}
    For a $d$-dimensional simplicial complex $(\SC, \pi)$ which is a $(\gamma_0, \ldots, \gamma_{d - 2})$-local spectral expander, it holds that for every $k = 1, \ldots, d$,
    \begin{align*}
        \lambda_2(\pdu_k) \le 1 - \frac{1}{k} \prod_{i = 0}^{k - 2} (1 - \gamma_i).
    \end{align*}
\end{theorem}

The following trickle-down theorem in~\cite{Oppenheim18} is the key ingredient.
\begin{lemma}[Oppenheim's trickle-down theorem~\cite{Oppenheim18}] \label{lem:trickling-down}
    For a $d$-dimensional pure simplicial complex $(\SC, \pi)$ and a facet $S \in \SC$, suppose that for every $v \in \SC_{S}(1)$, $\lambda_2(\pudnl_{S \cup \set{v}, 1}) \le \gamma < 1$ and $\pudnl_{S, 1}$ is connected.
    Then it holds that
    \[
        \lambda_2(\pudnl_{S, 1}) \le \frac{\gamma}{1 - \gamma}.
    \]
\end{lemma}

\subsection{Random regular bipartite graphs}
In this part, we introduce the model of random regular bipartite graphs.
Let $\RBG(n, \Delta)$ be the probabilistic model generating a graph from all $\Delta$-regular $2n$-vertex bipartite graphs uniformly at random.
For simplicity of analysis, we introduce the \emph{pairing model} to study random regular bipartite graphs.
\begin{definition}[Pairing model] \label{def:configuration-model}
    The pairing model generates a bipartite multi-graph $G = (V = V_L \cup V_R, E)$ in the following way:
    \begin{itemize}
        \item the vertex set is set by $V_L = V_R = [n]$;
        \item to generate $E$:
        \begin{itemize}
            \item define $\calV_L = V_L \times [\Delta]$ and $\calV_R = V_R \times [\Delta]$, and choose a uniformly random perfect matching between $\calV_L$ and $\calV_R$;
            \item for a pair $(u, i)$ and $(v, j)$ in the perfect matching, generate an edge $(u, v)$ in $E$.
        \end{itemize}
    \end{itemize}
\end{definition}
The following proposition builds a connection between $\RBG(n, \Delta)$ and $\PM(n, \Delta)$.
\begin{proposition}[\cite{Wormald99}] \label{prop:configuration-model-to-random-regular-bipartite-graph}
    For an event $\calE$, if $\Pr[\PM(n, \Delta)]{\calE} = o(1)$, then $\Pr[\RBG(n, \Delta)]{\calE} = o(1)$.
\end{proposition}

For a bipartite graph $G = (V = V_L \cup V_R, E)$, define the matrix $A_G \in \mathbb{R}^{V_L \times V_R}$ as
\[
    A_G(u, v) = \begin{cases}
        1 & (u, v) \in E \\
        0 & (u, v) \notin E
    \end{cases}\;.
\]
We remark here the matrix $A_G$ plays a significant role in the spectral analysis of $G$, since the adjacent matrix of $G$ can be denoted by
\begin{align*}
    \Adj(G) = \begin{pmatrix}
        0 & A_G \\
        A_G^{\top} & 0
    \end{pmatrix}\;.
\end{align*}
\begin{proposition}[\cite{KOGT26}] \label{prop:RBG-spectrum}
    With high probability over $G \sim \RBG(n, \Delta)$, $\lambda_2(A_G) \le 3\sqrt{\Delta}$.
\end{proposition}
We also use the following result.
\begin{proposition}[{\cite[Lemma 3.7]{KOGT26}}] \label{prop:neighbor-intersection}
    With high probability over $G = (V, E) \sim \RBG(n, \Delta)$, for every $u, v \in V$, $\abs{N_G(u) \cap N_G(v)} \le 2$.
\end{proposition}
\section{Ising Models on Random Regular Bipartite Graphs} \label{sec:random-regular-bipartite-Ising}

In this section, we capture the geometry of the configuration space of Ising models on random regular bipartite graphs.
The main challenge beyond the uniqueness threshold is the presence of strong dependencies among configurations. Instead of analyzing the full Gibbs measure directly, we partition the configuration space according to the number of occupied vertices on each side of the bipartite graph. Each part admits better expansion properties, which will later allow us to construct rapidly mixing Markov chains.

To describe the partition of configurations space $\Omega$, we introduce the concept of \emph{slices} in~\cite{KOGT26}.
For every natural number $s \ge 0$, we say a vertex subset $A \subseteq V_L$ of size $s$ is an $L$-side slice of size $s$, and say a vertex subset $B \subseteq V_R$ of size $s$ is an $R$-side slice of size $s$.
For two natural numbers $\ell, r \ge 0$, an $(\ell, r)$-two-side slice $S$ is a vertex subset of $V$ such that $\abs{S \cap V_L} = \ell$, $\abs{S \cap V_R} = r$ and $w(S) > 0$.
For simplicity, let $\Omega^{(\ell, r)}$ be the collection of all $(\ell, r)$-two-side slices.

We specify slices with their cardinalities.
For every natural number $s = 0, \ldots, n$, define the one-side partial partition function $\Z_{s}^{L}, \Z_{s}^{R}$ as
\begin{align*}
    \Z_{s}^{L} \defeq \sum_{A \subseteq V_L : \abs{A} = s} \sum_{B \subseteq V_R} w(A \cup B), \quad \Z_{s}^{R} \defeq \sum_{B \subseteq V_R : \abs{B} = s} \sum_{A \subseteq V_L} w(A \cup B).
\end{align*}
Similarly, for natural numbers $\ell = 0, 1, \ldots, n$ and $r = 0, 1, \ldots, n$, define the two-side partial partition function $\Z_{\ell, r}$ as
\begin{align*}
    \Z_{\ell, r} \defeq \sum_{S \in \Omega^{(\ell, r)}} w(S).
\end{align*}
The corresponding Gibbs distribution is denoted by
\begin{gather*}
    \mu_{s}^{L}(S) \defeq \mu\ab(S \mid S \subseteq V_L, \abs{S} = s), \quad
    \mu_{s}^{R}(S) \defeq \mu\ab(S \mid S \subseteq V_R, \abs{S} = s), \\
    \mu_{\ell, r}(S) \defeq \mu\ab(S \mid \abs{S \cap V_L} = \ell, \abs{S \cap V_R} = r).
\end{gather*}
For a vertex subset $A \subseteq V_L$, define
\begin{align*}
    W(A) \defeq \sum_{v \in V_R} \beta^{d_{A}(v)}.
\end{align*}
Similarly we can extend the definition of $W$ to every vertex subset of $V_R$.

Before stating properties for $\mu_{G; \lambda, \beta}$ when $G$ is generated from random regular bipartite graphs, we describe the main idea of the approximating algorithm to make our analysis more understandable.
To approximate $\Z_{G}(\lambda, \beta)$, we pick two proper parameters $\alpha, \rho \in (0, 1)$.
Then we define the quantity $\wh{Z}$ as
\begin{align*}
    \wh{Z} \defeq \sum_{\ell \le \alpha n, r \le \alpha n} \Z_{\ell, r} + \sum_{\alpha n < s \le \rho n} \Z_{s}^{L} + \sum_{\alpha n < s \le \rho n} \Z_{s}^{R}
\end{align*}
when $\alpha n < \rho n$ and
\begin{align*}
    \wh{Z} \defeq \sum_{\ell \le \rho n, r \le \rho n} \Z_{\ell, r}
\end{align*}
otherwise. The purpose of the one-side slices is to capture configurations where one side is moderately large while the other side remains small. \Cref{thm:approximation-to-partition-function} shows that configurations with both sides larger than $\alpha n$ have exponentially small contribution.

\begin{remark}
For a parameter $\theta \in (0, 1)$, we simply use $\theta n$ instead of $\floor{\theta n}$.
This takes little impact on our analysis but makes our statements much cleaner.
For simplicity when we are saying a fraction is in an interval, it means that we take values $\theta$ such that $\theta n$ is an integer.
\end{remark}

\subsection{One-side tail bound}
The first property we use is the \emph{tail bound} of $\mu$, which can be viewed as a generalization of~\cite[Lemma 1.11]{KOGT26}.
\begin{lemma} \label{lem:tail-bound}
    For a bipartite graph $G = (V = V_L \cup V_R, E)$ with $\abs{V_L} = \abs{V_R} = n$ and two real numbers $\lambda > 0, \beta \in [0, 1]$, the Ising model on $G$ at vertex interaction $\lambda$ and edge interaction $\beta$ satisfies that for every $\rho \in (0, 1)$,
    \begin{align*}
        \Pr[S \sim \mu]{\abs{S \cap V_L} > \rho n} \le \exp\ab(-n \KL{\rho}{\lambda/(1 + \lambda)}).
    \end{align*}
    The same tail bound holds for $V_R$ as well.
\end{lemma}
\begin{proof}
    For a vertex subset $A \subseteq V_L$, since random variables $\set{\id{v \in S}}_{v \in V_R}$ are independent conditioned on $S \cap V_L = A$, it holds that
    \begin{align*}
        \Pr{S \cap V_L = A} \propto \lambda^{\abs{A}} \prod_{v \in V_R} \ab(1 + \lambda \beta^{d_A(v)}).
    \end{align*}
    Therefore, for every $A\subseteq V_L$ and every $u \in V_L \setminus A$, since $\beta \le 1$,
    \begin{align*}
        \frac{\Pr{S \cap V_L = A \cup \set{u}}}{\Pr{S \cap V_L = A}} = \lambda \prod_{v \in V_R} \frac{1 + \lambda \beta^{d_{A \cup \set{u}}(v)}}{1 + \lambda \beta^{d_A(v)}} \le \lambda.
    \end{align*}
    Then we obtain that the random variable $X = \abs{S \cap V_L}$ is dominated by $n$ independent Bernoulli random variables $\Ber{\lambda/(1 + \lambda)}$.
    By the Chernoff bound~\Cref{lem:Chernoff-bound}, we conclude the lemma.
\end{proof}

The following lemma gives a proper approximation to $\Z_{G}(\lambda, \beta)$.
\begin{lemma} \label{thm:approximation-to-partition-function}
    For two parameters $\alpha, \rho \in (0, 1/2]$, assume that $\rho > \frac{\lambda}{1 + \lambda}$ and for every vertex subset $A \subseteq V_L$ of size $\alpha n$,
    \begin{align} \label{eq:one-side-weight-condition}
        \lambda W(A) \le (1 - \kappa) \abs{A}.
    \end{align}
    Define the quantity
    \begin{align*}
        \wh{Z} = \begin{cases}
            \sum_{\ell \le \alpha n, r \le \alpha n} \Z_{\ell, r} + \sum_{\alpha n < s \le \rho n} \Z_{s}^{L} + \sum_{\alpha n < s \le \rho n} \Z_{s}^{R} & \alpha \le \rho \\
             \sum_{\ell \le \rho n, r \le \rho n} \Z_{\ell, r} & \text{otherwise}
        \end{cases}\;.
    \end{align*}
    Then there is a constant $C = C(\kappa) = -\ln{(1 - \kappa)} - \kappa$ such that
    \begin{align*}
        \abs{\frac{\wh{Z}}{\Z_{G}(\lambda, \beta)} - 1} \le 2\e^{-n\KL{\rho}{\lambda/(1 + \lambda)}} + \e^{-C\alpha n}.
    \end{align*}
\end{lemma}
\begin{proof}
    When $\alpha > \rho$, by~\Cref{lem:tail-bound},
    \begin{align*}
        \abs{\frac{\wh{Z}}{\Z_{G}(\lambda, \beta)} - 1} &\le \Pr{\abs{S \cap V_L} > \rho n} + \Pr{\abs{S \cap V_R} > \rho n} \le 2\e^{-n\KL{\rho}{\lambda/(1 + \lambda)}}.
    \end{align*}
    When $\alpha \le \rho$, it holds that by~\Cref{lem:tail-bound},
    \begin{align*}
        \abs{\frac{\wh{Z}}{\Z_{G}(\lambda, \beta)} - 1} &\le \Pr{\abs{S \cap V_L} > \rho n} + \Pr{\abs{S \cap V_R} > \rho n} \\
        &+ \Pr{\alpha n < \abs{S \cap V_L} \le \rho n, \alpha n < \abs{S \cap V_R} \le \rho n} \\
        &\le 2\e^{-n\KL{\rho}{\lambda/(1 + \lambda)}} + \Pr{\alpha n < \abs{S \cap V_L} \le \rho n, \alpha n < \abs{S \cap V_R} \le \rho n}.
    \end{align*}
    For a vertex subset $A \subseteq V_L$ of size $> \alpha n$, conditional on $S \cap V_L = A$, it holds that the random variables $\set{\id{v \in S}}_{v \in V_R}$ are independent.
    Therefore, $\id{v \in S}$ is a Bernoulli random variable with rate
    \[
        p_v \defeq \Pr{v \in S \mid S \cap V_L = A} = \frac{\lambda \beta^{d_A(v)}}{1 + \lambda \beta^{d_A(v)}}.
    \]
    Then
    \[
        \E{\abs{S \cap V_R} \mid S \cap V_L = A} = \sum_{v \in V_R} p_v \le \lambda W(A).
    \]
    Pick a subset $B$ of $A$ of size $\alpha n$ and by monotonicity ($\beta \le 1$), $W(A) \le W(B)$.
    Then by Chernoff bound (\Cref{lem:Chernoff-bound}),
    \begin{align*}
        \Pr{\abs{S \cap V_R} > \alpha n \mid S \cap V_L = A} \le \exp(-\alpha n(-\ln{(1 - \kappa)} - \kappa)).
    \end{align*}
    By law of total expectation,
    \begin{align*}
        \mathbb{P}(\alpha n < \abs{S \cap V_L} &\le \rho n, \alpha n < \abs{S \cap V_R} \le \rho n) \\
        &\le \sum_{A \subseteq V_L : \abs{A} > \alpha n} \Pr{S \cap V_L = A} \Pr{\abs{S \cap V_R} > \alpha n \mid S \cap V_L = A} \\
        &\le \sum_{A \subseteq V_L : \abs{A} > \alpha n} \Pr{S \cap V_L = A} \e^{-\alpha n(-\ln{(1 - \kappa) - \kappa})} \\
        &\le \e^{-\alpha n(-\ln{(1 - \kappa) - \kappa})}.
    \end{align*}
    Substituting it into the upper inequality, we conclude that
    \begin{align*}
         \abs{\frac{\wh{Z}}{\Z_{G}(\lambda, \beta)} - 1} \le 2\e^{-n\KL{\rho}{\lambda/(1 + \lambda)}} + \e^{-\alpha(-\ln{(1 - \kappa)} - \kappa)n}
    \end{align*}
    which is the desired upper bound.
\end{proof}

\subsection{Concentration phenomena}
In this part, we reveal some concentration phenomena in the probability distribution $\mu_{G; \lambda, \beta}$.
For simplicity of analysis, we consider the pairing model $\PM(n, \Delta)$.
Recall the definition of $W(A)$,
By symmetry, the normalized expectation of $W(A)$ is defined by
\begin{align*}
    \mean(\theta) &= \frac{1}{n} \E[\PM(n, \Delta)]{W(A)} \\
    &= \frac{1}{n} \sum_{v \in V_R} \E[\PM(n, \Delta)]{\beta^{d_{A}(v)}} \\
    &= \frac{1}{\binom{\Delta n}{\Delta}} \sum_{j = 0}^{\Delta} \binom{\Delta \theta n}{j} \binom{\Delta(1 - \theta)n}{\Delta - j} \beta^j.
\end{align*}
Before all, we discuss the following approximation to $\mean(\theta)$ which could give some intuition to our discuss.
\begin{lemma}[Approximation to normalized expectation] \label{lem:mean-approximation}
    It holds that $\mean(\theta) = (1 - (1 - \beta)\theta)^{\Delta} + o_n(1)$.
\end{lemma}
\begin{proof}
    Let $X$ be a random variable denoting the number of marked elements in a $\Delta$-size sample drawn from $\Delta n$ elements containing $\Delta \theta n$ marked ones without replacement.
    Then by symmetry, $\mean(\theta) = \E{\beta^X}$.
    Let $Y$ be a random variable denoting the number of marked elements from $\Delta$ independent samples in the same population.
    We construct a coupling between $X$ and $Y$.
    For a process of $Y$, when the draws are distinct, we use it as $X$;
    otherwise, we resample $X$ from all distinct $\Delta$ samples uniformly at random.
    Then by union bound,
    \[
        \Pr{X \neq Y} \le \binom{\Delta}{2} \frac{1}{\Delta n} = \frac{\Delta - 1}{2n}.
    \]
    Since $\beta \in [0, 1]$,
    \[
        \E{\beta^X - \beta^Y} \le \Pr{X \neq Y} \le \frac{\Delta - 1}{2n}.
    \]
    It is clear that $Y \sim \Bin{\Delta, \theta}$.
    Then we conclude that
    \[
        \mean(\theta) = \E{\beta^Y} + o_n(1) = (1 - \theta + \theta \beta)^{\Delta} + o_n(1).
    \]
\end{proof}

We investigate the following concentration for $W(A)$. 
For simplicity, we provide the proof of it in~\Cref{sec:concentration-proof}.
\begin{restatable}{lemma}{oneSideWeightConcentration} \label{lem:one-side-weight-concentration}
    For a vertex subset $A \subseteq V_L$ with $\abs{A} \le n/2$, there exists an absolute constant $C > 0$ such that with high probability over the pairing model $\PM(n, \Delta)$, for every $t > 0$,
    \begin{align*}
        \Pr{\abs{W(A) - \E{W(A)}} \ge t} \le 2\exp\ab(-\frac{t^2}{2(1 - \beta)(2n + 2t/3)}).
    \end{align*}
    Specially, for the case $\beta = 1$, $W(A)$ is deterministic.
\end{restatable}

For every positive integer $n > 0$, define the function $h_n : [0, 1] \times \mathbb{R}_{> 0} \to \mathbb{R}$ as
\begin{align*}
    h_n(\theta, u) \defeq H(\theta) + \frac{u + 3\ln{(n + 1)}}{n}
\end{align*}
where $H(x) = -x\ln{x} - (1 - x)\ln{(1 - x)}$ is the binary entropy function.
For convenience, we set the error function as
\[
    \Error_{\beta}(\theta, u) = 4\ab((1 - \beta)h_n(\theta, u) + \sqrt{(1 - \beta)h_n(\theta, u)}).
\]
The following corollary, which is a direct result from~\Cref{lem:one-side-weight-concentration}, takes significant impact on our choice of the parameter $\alpha$.
\begin{corollary} \label{cor:weight-concentration}
    Fix a factor $u > 0$.
    With probability at least $1 - 2\e^{-u}$ over $\PM(n, \Delta)$, for every vertex subset $A \subseteq V_L$ or $A \subseteq V_R$ of size $\theta n \le n/2$,
    \begin{align} \label{eq:weight-concentration}
        \abs{\frac{W(A)}{n} - \mean(\theta)} \le \Error_{\beta}(\theta, u).
    \end{align}
    Moreover, taking $u = 2\ln{n}$,~\eqref{eq:weight-function} holds for every vertex subset $A \subseteq V_L$ or $A \subseteq V_R$ of size $\theta n \le n/2$ with probability $1 - o_n(1)$.
\end{corollary}
\begin{proof}
    When $\beta = 1$, $W(A)$ is a deterministic statistic and there remains nothing to do.
    For $\beta \in [0, 1)$, apply~\Cref{lem:tail-bound} with the choice
    \begin{align*}
        t = 4n\ab((1 - \beta)h_{n}(\theta, u) + \sqrt{(1 - \beta)h_n(\theta, u)}).
    \end{align*}
    It holds that for a subset $A$ of size $\theta n$, the failure probability is at most $2\e^{-2nh_n(\theta, u)}$.
    By the union bound and $\binom{n}{k} \le \e^{nH(k/n)}$ for $k \le n/2$, the total failure probability is at most
    \begin{align*}
        2 \sum_{k = 0}^{n/2} \binom{n}{k} 2\e^{-2nh_n(k/n, u)} &\le 4(n + 1)^{-6} \sum_{k = 0}^{n/2} \e^{-nH(k/n) - 2u} \le 2\e^{-u}.
    \end{align*}
    Therefore we conclude the lemma.
\end{proof}

Another important concentration phenomenon is the distribution of \emph{degree profiles}.
For a subset $A \subseteq V_L$ of size $\theta n$, its degree profile $p_A \in \mathbb{R}^{\Delta + 1}$ is a normalized vector defined as
\begin{align*}
    \forall j = 0, \ldots, \Delta, \quad p_A(j) = \frac{1}{n} \abs{\set{v \in V_R \mid d_{A}(v) = j}}
\end{align*}
and set its `expectation' profile as
\begin{align*}
    \forall j = 0, \ldots, \Delta, \quad \Phi_{\theta}(j) = \binom{\Delta}{j} \theta^j (1 - \theta)^{\Delta - j}.
\end{align*}
Let $\Profile(\theta)$ be all admissible degree profiles $p$ for a subset $A \subseteq V_L$ of size $\theta n$.
Define the residual factor $r_n(\cdot)$ as
\begin{align*}
    r_n(u) = \frac{1}{n}\ab((\Delta + 5) \ln{(n + 1)} + u + \ln{2(\Delta + 1)}).
\end{align*}
\begin{restatable}{proposition}{degreeProfileConcentration} \label{prop:degree-profile-concentration}
    For every factor $u > 0$, with probability $1 - \e^{-u}$ over $\PM(n, \Delta)$, for every vertex subset $A \subseteq V_L$ or $A \subseteq V_R$, it holds that
    \begin{align} \label{eq:degree-profile-concentration}
        \KL{p_A}{\Phi_{\theta}} \le H(\theta) + r_n(u)
    \end{align}
    where $\theta = \abs{A} / n$.
\end{restatable}
The proof of~\Cref{prop:degree-profile-concentration} is deferred to~\Cref{sec:concentration-proof}.
Hence for every real number $r \ge 0$, define
\[
    \Profile(\theta, r) \defeq \set{p \in \Profile(\theta) \;\colon\; \KL{p}{\Phi_{\theta}} \le H(\theta) + r},
    % \Profile(\theta, r) \defeq \set{p \in \mathbb{R}_{\ge 0}^{\Delta + 1} \;\colon\; \sum_{j = 0}^{\Delta} p(j) = 1, \sum_{j = 0}^{\Delta} jp(j) = \Delta \theta, \KL{p}{\Phi_{\theta}} \le H(\theta) + r},
\]
and
\[
    s_-(\theta, r) \defeq \inf_{p \in \Profile(\theta, r)} \sum_{j = 0}^{\Delta} p(j) \beta^j, \quad
    s_+(\theta, r) \defeq \sup_{p \in \Profile(\theta, r)} \sum_{j = 0}^{\Delta} p(j) \beta^j.
\]
Then consider the following bounds:
\[
    \sLower(\theta) \defeq \max\set{\mean(\theta) - \Error_{\beta}(\theta, u), s_-(\theta, r_n(u))}, \quad
    \sUpper(\theta) \defeq \min\set{\mean(\theta) + \Error_{\beta}(\theta, u), s_+(\theta, r_n(u))}.
\]
% By~\Cref{lem:one-side-weight-concentration,prop:degree-profile-concentration}, with probability at least $1 - 3\e^{-u}$, for every vertex subset $A \subseteq V_L$ or $A \subseteq V_R$ of size at most $n/2$,
% \begin{align} \label{eq:one-side-weight-approximation}
%     \sLower(\abs{A}/n) \le \frac{W(A)}{n} \le \sUpper(\abs{A}/n).
% \end{align}
\begin{corollary} \label{cor:one-side-weight-concentration}
    Fix $u > 0$.
    For the pairing model $G = (V = V_L \cup V_R, E) \sim \PM(n, \Delta)$, with probability at least $1 - 3\e^{-u}$, for every vertex subset $A \subseteq V_L$ or $A \subseteq V_R$ of size at most $n/2$,
    \[
        \sLower(\abs{A}/n) \le \frac{W(A)}{n} \le \sUpper(\abs{A}/n).
    \]
\end{corollary}
\begin{proof}
    The corollary is a direct result of~\Cref{lem:one-side-weight-concentration,prop:degree-profile-concentration}.
\end{proof}

The following lemma is also useful in our verfication of parameters.
\begin{lemma} \label{lem:useful-bound}
    On the event of~\Cref{cor:one-side-weight-concentration}, if $(1 - \beta)\theta\Delta \ge 6$, it holds that for all $A \subseteq V_L$ of size $\theta n$,
    \[
        \frac{W(A)}{n} \le \e^{-(1 - \beta)\theta\Delta/2} + \frac{3\ab(H(\theta) + r_n(u))}{(1 - \beta)\theta\Delta/2 - 1}.
    \]
\end{lemma}
\begin{proof}
    Consider the stochastic kernel $\mathscr{K}$ from $\set{0, \ldots, \Delta}$ to $\set{0, 1}$ defined by
    \[
        \mathscr{K}(j, 0) = 1 - \beta^j, \quad \mathscr{K}(j, 1) = \beta^j.
    \]
    Then by data processing,
    \[
        \KL{W(A)/n}{(1 - (1 - \beta)\theta)^{\Delta}} = \KL{p_A \mathscr{K}}{\Phi_{\theta} \mathscr{K}} \le H(\theta) + r_n(u).
    \]
    If $W(A)/n \le \e^{-(1 - \beta)\theta\Delta/2}$, there is nothing to prove.
    Otherwise,
    \[
        \KL{W(A)/n}{\e^{-(1 - \beta)\theta\Delta}} \le \KL{W(A)/n}{(1 - (1 - \beta)\theta)^{\Delta}} \le H(\theta) + r_n(u).
    \]
    By the inequality $\KL{W(A)/n}{\e^{-(1 - \beta)\theta\Delta}} \ge W(A)/n \ab((1 - \beta)\theta\Delta/2 - 1) - \e^{-(1 - \beta)\theta\Delta}$, we conclude the lemma.
\end{proof}
\section{Trickle-down Theorems for Slices} \label{sec:trickle-theorems}

The previous section shows that typical slices of the configuration space have
a regular geometric structure. In this section, we exploit this structure to
prove rapid mixing of local walks.

We state trickle-down theorems for two-side slices and one-side ones respectively.
Set parameters $\eta \in (0, 1)$ and $C_{\TD} > 4$.
Given $\lambdaUpperBound$ and $\ParameterCondition$ for some $\lambda_0$ such that $(1 + \eta)\lambda_0/(1 + \lambda_0) \le 1/2$ and $C > 0$, we take
\begin{align} \label{eq:upper-bound-choice}
    \rho_{*} \defeq (1 + \eta) \frac{\lambda}{1 + \lambda}
\end{align}
and
\begin{align} \label{eq:overlap-choice}
    \alpha_{*} \defeq \max\set{\theta \in [0, \rho_{*}] : \sLower(\theta) - \theta \ge C_{\TD} (1 - \beta) \cdot 3\sqrt{\Delta} \theta}.
\end{align}
\begin{remark}
    Note that $\alpha_{*}$ is always a valid value and we can further assume that $\alpha_* n$ is a integer.
\end{remark}

\subsection{Two-side trickle-down theorem}
First we discuss the trickle-down theorem for two-side slices.
For integers $\ell = 0, \ldots, n$ and $r = 0, \ldots, n$, the simplicial complex $(\SC_{\ell, r}, \mu_{\ell, r})$ induced by $\mu_{\ell, r}$ is constructed as:
\begin{itemize}
    \item $\SC_{\ell, r}$ is the downward closed collection of $\Omega^{(\ell, r)}$;
    \item the weight function of $\SC$ is defined as~\eqref{eq:simplicial-complex-weight-function} with $\pi = \mu_{\ell, r}$.
\end{itemize}
\begin{lemma} \label{lem:two-side-trickle-down}
    Suppose that the statements in~\Cref{prop:RBG-spectrum,cor:one-side-weight-concentration} holds.
    Then for every $\ell \le \alpha_{*}n$ and $r \le \alpha_{*}n$, the facet $S$ in $(\SC = \SC_{\ell, r}, \pi = \mu_{\ell, r})$ induced by $\mu_{\ell, r}$ of codimension $2$ satisfies that
    \begin{align*}
        \lambda_2(\pudnl_{S, 1}) \le \frac{1}{2(\ell + r)}.
    \end{align*}
    for all sufficiently large positive integer $n$.
\end{lemma}
\begin{proof}
    When $\beta = 1$, the result holds trivially.
    We suppose that $\beta < 1$.
    For every facet $S$ of codimension $2$, let $A = S \cap V_L$ and $B = S \cap V_R$.
    If $\abs{A} = \ell$ or $\abs{B} = r$, the link of $S$ is a complete graph with positive vertex weights.
    It is not hard to see $\lambda_2(\pudnl_{S, 1}) \le 0$.
    When $\abs{A} = \ell - 1$ and $\abs{B} = r - 1$,
    \begin{align*}
        \pudnl_{S, 1} = \begin{pmatrix}
            0 & P \\
            Q & 0
        \end{pmatrix}
    \end{align*}
    where $P \in \mathbb{R}^{(V_L \setminus A) \times (V_R \setminus B)}, Q \in \mathbb{R}^{(V_R \setminus B) \times (V_L \setminus A)}$ are defined by
    \begin{align*}
        P(u, v) &= \frac{\beta^{d_A(v)} \beta^{\id{(u, v) \in E}}}{ \sum_{z \in V_R \setminus B} \beta^{d_A(z)} \beta^{\id{(u, z) \in E}}}, \\
        Q(v, u) &= \frac{\beta^{d_B(u)} \beta^{\id{(u, v) \in E}}}{ \sum_{z \in V_L \setminus A} \beta^{d_B(z)} \beta^{\id{(z, v) \in E}}}.
    \end{align*}
    For each $u \in V_L \setminus A$ and $v \in V_R \setminus B$, define
    \[
        p_u = \beta^{d_B(u)}, \quad q_v = \beta^{d_A(v)}.
    \]
    Consider the matrix $M \in \mathbb{R}^{(V_L \setminus A) \times (V_R \setminus B)}$
    \[
        M(u, v) \defeq p_u q_v \beta^{\id{(u, v) \in E}} = p_u(1 - (1 - \beta)A_G(u, v))q_v.
    \]
    Moreover, define diagonal matrices $D_p = \diag{p_u : u \in V_L \setminus A}$, $D_q = \diag{p_v : v \in V_R \setminus B}$ and
    \begin{align*}
        \forall u \in V_L \setminus A, \quad D_L(u, u) &= p_u \sum_{v \in V_R \setminus B} q_v (1 - (1 - \beta) A_G(u, v)), \\
        \forall v \in V_R \setminus B, \quad D_R(v, v) &= q_u \sum_{u \in V_L \setminus A} p_u (1 - (1 - \beta) A_G(u, v)).
    \end{align*}
    We represent $\pudnl_{S, 1}$ in the form of
    \begin{align*}
        \begin{pmatrix}
            0 & D_{L}^{-1} M \\
            D_R^{-1} M^{\top} & 0
        \end{pmatrix}
    \end{align*}
    and $M = D_p(\one_{V_L \setminus A} \one_{V_R \setminus B}^{\top} - (1 - \beta)A_G) D_q$.
    By mononicity of $\beta^{d_A(v)}$, the choice of $\alpha_{*}$ and~\Cref{cor:one-side-weight-concentration}, it holds that
    \begin{align*}
        \sum_{v \in V_R \setminus B} q_v(1 - (1 - \beta) A_G(u, v)) \ge n\sLower(\alpha_{*}) - \alpha_{*} n\ge C_{\TD}(1 - \beta) 3\sqrt{\Delta} \alpha_{*}n - \Delta.
    \end{align*}
    Similar bound holds for $\sum_{u \in V_L \setminus A} p_u(1 - (1 - \beta)A_G(u, v))$.
    Therefore, by~\Cref{prop:RBG-spectrum}, it holds that
    \begin{align*}
        \lambda_2(\pudnl_{S, 1}) &\le \frac{(1 - \beta)\lambda_2(A_G)}{C_{\TD}(1 - \beta) 3\sqrt{\Delta} (\alpha_{*}n - \Delta)} \\
        &\le \frac{1 + o_n(1)}{C_{\TD} \alpha_{*}n} \\
        &\le \frac{1}{2(\ell + r)}
    \end{align*}
    for all sufficiently large positive integer $n$.
\end{proof}

\subsection{One-side trickle-down theorem}
Now we discuss the trickle-down theorem on one-side slices.
When $\beta = 0$,~\cite{KOGT26} proposes a way to derive the trickle-down theorem.
We show how to extend their method to adapt the case $\beta > 0$.

By symmetry, we only need to consider the left-side slices $\mu_{s}^{L}$.
Fix a size $s = 0, \ldots, n$.
Construct the simplicial complex $(\SC = \SC_s^L, \pi = \mu_s^L)$ as
\begin{itemize}
    \item $\SC_s^L$ is the downward closed collection of $\binom{V_L}{s}$;
    \item the weight function of $\SC$ is defined as~\eqref{eq:simplicial-complex-weight-function} with $\pi = \mu_s^L$.
\end{itemize}
Recall that for every vertex subset $A \in \binom{V_L}{s}$,
\[
    \mu_{s}^{L}(A) \propto W(A) = \prod_{v \in V_R} \ab(1 + \lambda \beta^{d_{A}(v)}).
\]
Define the sequences $\set{\varphi_j}_{j = 0}^{\Delta}, \set{\phi_j}_{j = 0}^{\Delta - 1}$ and $\set{\delta_j}_{j = 0}^{\Delta - 2}$ as
\[
    \varphi_j \defeq 1 + \lambda \beta^j, \quad
    \phi_j \defeq \ln\frac{\varphi_{j}}{\varphi_{j + 1}}, \quad
    \delta_j \defeq \frac{\varphi_j \varphi_{j + 2}}{\varphi_{j + 1}^2} - 1 = \frac{\lambda(1 - \beta)^2 \beta^j}{(1 + \lambda \beta^{j + 1})^2}.
\]
Now given a facet $A$ of codimension $2$, for a vertex $x \in V_L \setminus A$, define
\begin{align*}
    q_A(x) = \prod_{v \in N_G(x)} \frac{\varphi_{d_A(v) + 1}}{\varphi_{d_A(v)}}.
\end{align*}
We define the off-diagonal matrix $K \in \mathbb{R}^{(V_L \setminus A) \times (V_L \setminus A)}$ as
\begin{align*}
    \forall x \neq y \in V_L \setminus A, \quad K(x, y) = \prod_{v \in N_G(x) \cap N_G(y)} (1 + \delta_{d_A(v)})
\end{align*}
with convention that the empty product is $1$.
Then for every $x \in V_L \setminus A$, set
\[
    Z_{A}(x) \defeq \sum_{y \in V_L \setminus A} q_A(y) K(x, y).
\]
Therefore, a direct calculation gives that for $x \neq y \in V_L \setminus A$,
\begin{align*}
    \pudnl_{S, 1}(x, y) = \frac{q_A(y) K(x, y)}{Z_{A}(x)}.
\end{align*}
In addition, the average free energy of $p_A$ is defined as
\begin{align*}
    \averageEnergy_A = -\frac{1}{n - \abs{A}} \sum_{x \in V_L \setminus A} \ln{q_A(x)} = \frac{1}{n - \abs{A}} \sum_{j = 0}^{\Delta} n p_A(j) (\Delta - j) \phi_j.
\end{align*}

\begin{lemma} \label{lem:one-side-trickle-down}
    Assume that statements in~\Cref{prop:RBG-spectrum,prop:neighbor-intersection} holds.
    Set
    \[
        U_G(\lambda, \beta) = \lambda(1 - \beta)^{2} \lambda_2(A_G)^2 + \lambda^2(1 - \beta)^4 - 1.
    \]
    If $U_G(\lambda, \beta) \le 0$, it holds that $\lambda_2(\pudnl_{S, 1}) \le 0$.
    Otherwise,
    \[
        \lambda_2(\pudnl_{S, 1}) \le \frac{U_G(\lambda, \beta) \e^{\averageEnergy_A}}{n - \abs{A} - \e^{\overline{L}_A}}.
    \]
\end{lemma}
\begin{proof}
    Define a matrix $H \in \mathbb{R}^{(V_L \setminus A) \times V_R}$ as $H(x, v) = \id{(x, v) \in E}$.
    Set the diagonal matrix $D_{\delta} = \diag{\delta_{d_A(v)} : v \in V_R}$ and an off-diagonal matrix $M_2 \in \mathbb{R}^{(V_L \setminus A) \times (V_L \setminus A)}$ as
    \begin{align*}
        \forall x \neq y \in V_L \setminus A, \quad M_2(x, y) = \begin{cases}
            \delta_{d_A(v_1)} \delta_{d_A(v_2)}, & N_G(x) \cap N_G(y) = \set{v_1, v_2} \\
            0 & \text{otherwise}
        \end{cases}\;.
    \end{align*}
    By~\Cref{prop:neighbor-intersection}, it holds that
    \[
        K = \one_{V_L \setminus A} \one_{V_L \setminus A}^{\top} + H D_{\delta} H^{\top} - \diag{1 + \sum_{v \in N_G(x)} \delta_{d_A(v)} : x \in V_L \setminus A} + M_2.
    \]
    Recall the definition of $\set{\delta_j}_{j = 0}^{\Delta - 2}$, it always holds that $\delta_j \le \lambda(1 - \beta)^2$.
    Therefore $D_{\delta} \preceq \lambda(1 - \beta)^2 \ID$.
    For $M_2$, note that $M_2$ is symmetric and of the form
    \[
        \begin{pmatrix}
            0 & \ab(m_{xy})_{x, y} \\
            \ab(m_{xy})_{x, y}^{\top} & 0
        \end{pmatrix}, \quad m_{xy} \le \lambda^2(1 - \beta)^4.
    \]
    Therefore $M_2 \preceq \lambda^2(1 - \beta)^4 \ID$.
    Consequently,
    \[
        K \preceq \one_{V_L \setminus A} \one_{V_L \setminus A}^{\top} + \lambda(1 - \beta)^2 H H^{\top} + (\lambda^2(1 - \beta)^4 - 1)\ID.
    \]
    Note that $H$ is a principal compression of $A_G$.
    Therefore by regularity,
    \begin{align*}
        H H^{\top} \preceq A_G A_G^{\top} \preceq \frac{\Delta^2}{n} \one_{V_L} \one_{V_R}^{\top} + \lambda_2(A_G) \ID.
    \end{align*}
    Therefore
    \[
        K \preceq \ab(1 + \frac{\lambda(1 - \beta)^2 \Delta^2}{n}) \one \one^{\top} + U_G(\lambda, \beta) \ID.
    \]
    A similar argument in~\cite{KOGT26} subtracts the rank-one part for $\pudnl_{S, 1}$.
    Hence we only need to give bounds to $q_A(x)$ and $Z_A(x)$.
    For $q_A(x)$, it is clear to see $q_A(x) \le 1$ since $\beta \le 1$.
    For $Z_A(x)$, by Jensen's inequality,
    \begin{align*}
        Z_{A}(x) &\ge \sum_{y \neq x \in V_L \setminus A} q_A(y) \\
        &\ge \sum_{y \in V_L \setminus A} q_A(y) - 1 \\
        &\ge (n - \abs{A}) \e^{-\averageEnergy_A} - 1.
    \end{align*}
    Hence $1/Z_{A}(x) \le \e^{\averageEnergy_A}/(n - \abs{A} - \e^{\averageEnergy_A})$.
    Combining all things together, we conclude the lemma.
\end{proof}

\section{Design of Our Approximate Algorithm} \label{sec:algorithm-design}
In this section, we give our algorithm to approximate $\Z_{G}(\lambda, \beta)$ on random regular bipartite graphs and therefore prove~\Cref{thm:random-regular-bipartite-Ising-FPRAS}.

Given $\lambdaUpperBound$ and $\ParameterCondition$ for some $\lambda_0$ and $C$, we show how to set parameters to apply the above approximation and trickle-down theorems.
For some sufficiently large constant $K > 0$, define $\alpha_0$ be the unique solution of the following equation on $[0, 1/2]$:
\[
    (1 - (1 - \beta) \alpha)^{\Delta} = K\ab(1 + 3(1 - \beta)\sqrt{\Delta}) \alpha.
\]
Pick
\begin{align} \label{eq:parameter-settings}
    \rho = (1 + \eta) \frac{\lambda}{1 + \lambda}, \quad \alpha = \min\set{\alpha_0, \rho}.
\end{align}
\begin{remark}
    Given $K > 0$, we only need a largest $\alpha' \le \alpha_0$ such that $\alpha'n$ is an integer.
    Therefore it costs polynomial time to determine the value of $\alpha'$ and takes little impact on our analysis.
    For simplicity, we employ the solution $\alpha_0$ in following proofs.
\end{remark}

\begin{lemma} \label{lem:clean-weight-concentration}
    Fix a sufficiently large constant $K > 0$.
    For $\beta \in (0, 1)$ and sufficiently large factor $\Delta > 0$, it holds that
    \begin{align*}
        H(\alpha_0) \le \frac{(1 + 3(1 - \beta) \sqrt{\Delta})^2 \alpha_0^2}{1 - \beta}.
    \end{align*}
    Then on the event of~\Cref{cor:one-side-weight-concentration}, for every $\theta \le \alpha_0$,
    \[
        \frac{3}{4}(1 - (1 - \beta)\theta)^{\Delta} \le \frac{W(A)}{n} \le \frac{5}{4} (1 - (1 - \beta)\theta)^{\Delta}.
    \]
\end{lemma}
\begin{proof}
    Take
    \[
        y = (1 - \beta)\Delta \alpha_0, \quad t = (1 - \beta)\sqrt{\Delta}, \quad Q = \frac{(1 + 3(1 - \beta)\sqrt{\Delta})^2 \alpha_0}{1 - \beta}.
    \]
    It is sufficient to show that $\ln{(\e/\alpha_0)} \le Q$ and by the inequality $H(\alpha_0) \le \alpha_0 \ln{(\e/\alpha_0)}$ we conclude the lemma.

    Recall the choice $\alpha_0$ and we calculate:
    \[
        \ln{(1/\alpha_0)} = \ln{\ab(K\ab(1 + 3(1 - \beta)\sqrt{\Delta}))} - \Delta \ln{\ab(1 - \frac{y}{\Delta})}.
    \]
    A standard estimation to $-\ln{(1 - x)}$ on interval $[0, 1)$ and $y/\Delta \le \alpha_0 \le 1/2$ gives
    \[
        y \le -\Delta \ln{\ab(1 - \frac{y}{\Delta})} \le \frac{y}{1 - y/\Delta} \le 2y.
    \]
    Therefore
    \[
         \ln{\ab(K\ab(1 + 3(1 - \beta)\sqrt{\Delta}))} + y \le \ln{(1/\alpha_0)} \le  \ln{\ab(K\ab(1 + 3(1 - \beta)\sqrt{\Delta}))} + 2y.
    \]

    When $3(1 - \beta)\sqrt{\Delta} \ge 1$, it holds that $t \ge 1/3$ and
    \[
        3t \le 1 + 3(1 - \beta)\sqrt{\Delta} \le 6t, \quad Q \ge 9y.
    \]
    By the equation determining $\alpha_0$, we can also have the following inequalities:
    \begin{align*}
        \begin{cases}
        \e^{-2y} \le K(1 + 3(1 - \beta)\sqrt{\Delta}) \frac{y}{(1 - \beta)\Delta} \\
        \sqrt{\Delta} \le 6Ky\e^{2y}
        \end{cases}\;.
    \end{align*}
    By a crude inequality $1 + 3(1 - \beta)\sqrt{\Delta} \le 4\sqrt{\Delta}$, we upper bound the term $\ln{(1 + 3(1 - \beta)\sqrt{\Delta})}$ as
    \begin{align*}
        \ln{(1 + 3(1 - \beta)\sqrt{\Delta})} \le \ln{4} + \frac{1}{2}\ln{\Delta} \le \ln{24K} + 3y
    \end{align*}
    and we obtain an additional bound $y \ge (1 + 2\ln{K} + \ln{24})$ if
    \[
        \Delta \ge 36K^2(1 + 2\ln{K} + \ln{24})^2 \e^{4(1 + 2\ln{K} + \ln{24})}.
    \]
    Therefore
    \[
        \ln{(\e/\alpha_0)} \le 1 + 2\ln{K} + \ln{24} + 5y \le 6y \le Q.
    \]

    When $3(1 - \beta)\sqrt{\Delta} < 1$, then it holds that $t < 1/3$ and thus $1 + 3(1 - \beta)\sqrt{\Delta} < 2$.
    We bound $Q$ as
    \[
        Q = \ab(\frac{1 + 3(1 - \beta)\sqrt{\Delta}}{t})^2 y > 9y.
    \]
    The term $\ln{(\e/\alpha_0)}$ can be upper bounded by $1 + \ln{2K} + 2y$.
    Together with $(1 - \beta)\Delta \le 2Ky\e^{2y}$ by the equation determining $\alpha_0$, if $(1 - \beta)\Delta \ge L$ where $L = 2K(1 + \ln{2K}) \e^{2(1 + \ln{2K})}$, it holds that
    \[
        \ln{(\e/\alpha_0)} \le 3y \le Q.
    \]
    Otherwise, when $\Delta \ge 2L$, it holds that
    \[
        K(1 + 3(1 - \beta)\sqrt{\Delta}) \alpha_0 = (1 - (1 - \beta)\alpha_0)^{\Delta} \ge (1 - L/\Delta)^{\Delta} \ge \e^{-2L}.
    \]
    Thus $\alpha \ge \alpha_K = \e^{-2L}/(K(1 + 3(1 - \beta)\sqrt{\Delta}))$.
    For sufficiently large $\Delta$ depending only on $K$, it holds that
    \[
        \ln{(\e/\alpha_0)} \le \ln{(\e/\alpha_K)} \le \frac{\alpha_K}{1 - \beta} \le Q.
    \]
    Therefore we prove that $H(\alpha_0) \le \alpha_0 Q$.

    On the event of~\Cref{cor:one-side-weight-concentration}, observe that we only need to show the case $\theta = \alpha_0$.
    Pick $K$ as a sufficient large positive number and $\Delta$ is sufficiently large with respect to $K$, $H(\alpha_0)$ is at most $(1 - (1 - \beta)\alpha_0)^{\Delta}/4$ by the equation defining $\alpha_0$.
    Thus we conclude the full lemma.
\end{proof}

The following verification for $\alpha$ and $\rho$ allows us to apply the approximation to the partition function $\Z_{G}(\lambda, \beta)$ (\Cref{thm:approximation-to-partition-function}) and trickle-down theorems (\Cref{lem:two-side-trickle-down,lem:one-side-trickle-down}) safely.
\begin{proposition} \label{prop:parameter-verfication-alpha}
    For sufficiently large constant $K > 0$, recall that the parameter $\alpha_0 \in [0, 1/2]$ is defined by
    \[
        (1 - (1 - \beta)\alpha_0)^{\Delta} = K\ab(1 + 3(1 - \beta)\sqrt{\Delta}) \alpha_0
    \]
    and $\alpha = \min\set{\alpha_0, \rho}$.
    There exist constants $\lambda_0 > 0$ and $C > 0$ such that on the event of all results in~\Cref{lem:clean-weight-concentration} and the assumption $\lambdaUpperBound$ and $\ParameterCondition$, it holds that $\alpha \le \alpha_*$ and at least one of the following holds:
    \begin{enumerate}
        \item $\alpha = \rho$; or
        \item for all $A \subseteq V_L$ of size $\alpha n$, it holds that
        \[
            \lambda W(A) \le \frac{1}{2} \abs{A}.
        \]
    \end{enumerate}
\end{proposition}
\begin{proof}
    We show that
    \[
        \lambda \sUpper(\alpha) \le \frac{1}{2} \alpha
    \]
    and thus by~\Cref{cor:one-side-weight-concentration},
    \[
        \lambda W(A) \le \lambda n \sUpper(\alpha) \le \frac{1}{2} \abs{A}.
    \]
    To verify this, note that for all $\theta \le \alpha_0$, it holds that $(1 - (1 - \beta) \theta)^{\Delta} \ge K(1 + 3(1 - \beta)\sqrt{\Delta}) \theta$.1
    By choosing $K$ large enough, we can make $\alpha \le \alpha_{*}$.

    When $\alpha = \rho$, there is nothing to do.
    Otherwise, by~\Cref{lem:clean-weight-concentration}, we obtain that
    \[
        \frac{\lambda \sUpper(\alpha)}{\alpha} \le \ab(\frac{5}{4} + o_n(1)) \lambda K(1 + 3(1 - \beta)\sqrt{\Delta}) \le 2K (\lambda + 3\lambda(1 - \beta)\sqrt{\Delta}).
    \]
    Choose $\lambda_0$ and $C$ as small constants to make $2K(\lambda_0 + C) \le 1/2$.
    Then we conclude the proposition.
\end{proof}

\begin{proposition} \label{prop:parameter-verify-rho}
    For sufficiently large constant $K$ as in~\Cref{lem:clean-weight-concentration}, there exist $\lambda_0, C > 0$, a positive integer $\Delta_0$ such that the following holds.
    Fix $\Delta \ge \Delta_0$ and
    \[
        \lambdaUpperBound, \quad \ParameterCondition.
    \]
    For all sufficiently large positive integer $n$, we take $u = 2\ln{n}$ and suppose the events of~\Cref{prop:RBG-spectrum,prop:neighbor-intersection},~\Cref{cor:one-side-weight-concentration,prop:degree-profile-concentration} hold.
    In addition, assume that $\floor{\alpha n} < \floor{\rho n}$ (otherwise, there is nothing to do).
    Then for every $s \in \set{\floor{\alpha n} + 1, \ldots, \floor{\rho n}}, s \ge 2$, it holds that for every facet $A \in \binom{V_L}{s - 2}$ of codimension $2$, by setting $\theta = \abs{A}/n$,
    \[
        \theta \cdot \max\set{U_G(\lambda, \beta), 0} \cdot \e^{\averageEnergy_{A}} \le \begin{cases}
            O\ab(C\frac{\ln{\Delta}}{\Delta^{3/8}}), & (1 - \beta)\theta \Delta \le \ln{\Delta} \\
            O(C^2), & (1 - \beta)\theta \Delta \ge \ln{\Delta}
        \end{cases}\;.
    \]
    and therefore,
    \[
        1 - \theta - \frac{\e^{\averageEnergy_A}}{n} > 0, \quad \frac{2(\theta + 2/n) \e^{\averageEnergy_A} \max\set{U_G(\lambda, \beta), 0}}{1 - \theta - \e^{\averageEnergy_A}/n} \le \frac{1}{2}.
    \]
    and thus
    \[
        \lambda_2(\pudnl_{A, 1}) \le \frac{1}{2s}.
    \]
\end{proposition}
\begin{proof}
    Since $\beta \in (0, 1)$, it holds that
    \[
        0 \le \phi_j \le \lambda(1 - \beta) \beta^j.
    \]
    By definition of $\averageEnergy_A$, it holds that
    \[
        \averageEnergy_A \le \frac{\lambda(1 - \beta) \Delta}{1 - \theta} \frac{W(A)}{n}.
    \]
    To bound $W(A)$, we extend $A$ to $A'$ by adding two extra vertices and thus $\abs{A'} \ge \alpha n$.
    This operation changes the value of $W(A)$ by at most $2(1 - \beta)\Delta$, meaning that
    \[
        W(A) \le W(A') + 2(1 - \beta) \Delta.
    \]
    When $(1 - \beta) \theta \Delta \le \ln{\Delta}$, by~\Cref{lem:clean-weight-concentration},
    \[
        \averageEnergy_A \le 4K\lambda(1 - \beta)\Delta(1 + 3(1 - \beta)\sqrt{\Delta})\alpha_0 \le 12K(\lambda + \lambda(1 - \beta)\sqrt{\Delta}) \ln{\Delta}.
    \]
    Pick $\lambda_0$ and $C$ as sufficiently small constants such that $12K(\lambda_0 + C) \le 1/8$.
    Therefore $\e^{\averageEnergy_{A}} \le \Delta^{1/8}$.
    In addition, by~\Cref{prop:RBG-spectrum}, $\max\set{U_G(\lambda, \beta), 0} \le 100\lambda(1 - \beta)^2 \Delta$.
    Therefore,
    \begin{align} \label{eq:parameter-rho}
        \theta \e^{\averageEnergy_{A}} \max\set{U_G(\lambda, \beta), 0} \le 100 \lambda(1 - \beta) (1 - \beta)\theta \Delta \e^{1/8} \le \frac{100C\ln{\Delta}}{\Delta^{3/8}}.
    \end{align}

    When $(1 - \beta)\theta\Delta > \ln{\Delta}$, by~\Cref{lem:useful-bound} and $\theta \le 1/2$,
    \begin{align*}
        \averageEnergy_A &\le \frac{\lambda(1 - \beta)\Delta}{1 - \theta} \ab(\e^{-(1 - \beta)\theta\Delta/2} + \frac{3(H(\theta) + r_n(u))}{(1 - \beta)\theta\Delta/2 - 1}) \\
        &\le 2\lambda(1 - \beta)\Delta \e^{-(1 - \beta)\theta\Delta/2} + \frac{6\lambda(1 - \beta)\Delta(H(\theta) + r_n(u))}{(1 - \beta)\theta\Delta/2 - 1} \\
        &\le 2\lambda(1 - \beta)\sqrt{\Delta} + 18\lambda \frac{H(\theta) + r_n(u)}{\theta}.
    \end{align*}
    We take $u = 2\ln{n}$ and obtain that $r_n(u) = O_{\Delta}(\ln{n}/n)$.
    Then we have
    \[
        \averageEnergy_A \le 2C + 18\lambda\ln{(\e/\theta)} + o_n(1).
    \]
    Take $\lambda_0$ and $C$ sufficiently small.
    Since $\theta \le \rho = O(\lambda)$, we can obtain that
    \begin{align*}
        \theta \e^{-\averageEnergy_A} \max\set{U_G(\lambda, \beta), 0} &\le 10000\lambda(1 - \beta)^2 \Delta \theta^{1 - 10000\lambda} \\
        &\le O\ab(\lambda^2(1 - \beta)^2 \Delta \lambda^{-10000\lambda}) \\
        &\le O(C^2).
    \end{align*}
    To show~\eqref{eq:parameter-rho}, we set $n_0$ sufficient large and for all $n \ge n_0$, the full lemma holds.
\end{proof}

Now we are ready to design our algorithm and thus prove~\Cref{thm:random-regular-bipartite-Ising-FPRAS}.
\begin{proof}[Proof of~\Cref{thm:random-regular-bipartite-Ising-FPRAS}]
    When $\beta = 0$, employ the main result of~\cite{KOGT26}.
    When $\beta = 1$, it is direct to see
    $$
        \Z_G(\lambda, 1) = (1 + \lambda)^{\abs V}
    $$
    and we can calculate it precisely in polynomial time.

    Assume that $\beta \in (0, 1)$.
    Take $u = 2\ln{n}$.
    Assume that~\Cref{prop:RBG-spectrum,prop:neighbor-intersection},~\Cref{cor:one-side-weight-concentration,prop:degree-profile-concentration} hold.
    By union bound, these events hold with probability $1 - o_n(1)$.
    Pick $\alpha$ and $\rho$ as~\eqref{eq:parameter-settings}.
    By~\Cref{prop:parameter-verfication-alpha,prop:parameter-verify-rho},~\Cref{thm:approximation-to-partition-function},~\Cref{lem:two-side-trickle-down} and~\Cref{lem:one-side-trickle-down} holds.
    When $\eps \le 4(2\e^{-n\KL{\rho}{\lambda/(1 + \lambda)}} + \e^{-\alpha(\ln{2} - 1/2)n})$, we just enumerate all configurations by brute force and calculate the total weights.
    This terminates in time $\poly{n, \eps^{-1}}$.

    Otherwise, when $\floor{\alpha n} = \floor{\rho n}$, we set $\wh{Z} \defeq \sum_{\ell \le \rho n, r \le \rho n} \Z_{\ell, r}$ and for a fixed $\ell, r \le \rho n$, by~\Cref{lem:two-side-trickle-down,lem:trickling-down} and thus~\Cref{thm:local-to-global}, there is a rapid sampler to sample from $\mu_{\ell, r}$.
    However, note that the instance is not self-reducible.
    Instead, by a similar argument in~\cite{KOGT26}, the random walks on the simplicial complex is self-reducible.
    Therefore, there is a randomized algorithm to output an $(\eps/5)$-approximation to $\Z_{\ell, r}$ in time $\poly{n, \eps^{-1}}$, leading to an $(\eps/5)$-approximation to $\wh{Z}$.
    Thus we can design an $\FPRAS$ for $\Z_{G}(\lambda, \beta)$.

    When $\floor{\alpha n} = \floor{\rho n}$, we set
    \[
        \wh{Z} \defeq \sum_{\ell \le \alpha n, r \le \alpha n} \Z_{\ell, r} + \sum_{\alpha n < s \le \rho n} \Z_{s}^{L} + \sum_{\alpha n < s \le \rho n} \Z_{s}^{R}
    \]
    By above, we only need to design an efficient sampler for $\mu_{\ell, r}$, $\mu_{s}^{L}$ and $\mu_{s}^R$.
    By~\Cref{prop:parameter-verfication-alpha,prop:parameter-verify-rho}, we apply~\Cref{lem:two-side-trickle-down,lem:one-side-trickle-down} and thus by~\Cref{lem:trickling-down,thm:local-to-global}, for every $\ell \le \alpha n, r \le \alpha n$ and $s \le \beta$, the down-up walks on all of $\mu_{\ell, r}, \mu_{s}^{L}$ and $\mu_{s}^R$ exhibit polynomial-time mixing.
    Therefore, we can output an $(\eps/5)$-approximation to $\wh{Z}$ in time $\poly{n, \eps^{-1}}$ with probability at least $3/4$, which is an $\eps$-approximation to $\Z_{G}(\lambda, \beta)$.
\end{proof}

\bibliographystyle{alpha}
\bibliography{refs}

\appendix
\section{Equivalence of Models} \label{sec:model-equivalence}
We discuss here the equivalence between subgraph-world models and spin systems on regular graphs.
Fix a $\Delta$-regular graph $G = (V, E)$.
For the physical Ising model on $G$ with inverse temperature $J$ and external field $h$, define the partition function
$$
    \Z_G^\spin(h, J) = \sum_{\sigma \in \set{-, +}^V} \exp\ab(-J\sum_{(u, v) \in E} \sigma_u \sigma_v + h \sum_{v \in V} \sigma_v).
$$
Since $G$ is a $\Delta$-regular graph, it holds that
$$
    \Z_G^\spin(h, J) = \e^{-J\Delta n - 2hn} \Z_G(\lambda, \beta)
$$
where $\beta = \e^{-4J}$ and $\lambda = \e^{2J + 2h}$.
Therefore it suffice to consider the Ising model defined by the subgraph world when the underlying graph is regular.

\section{Proofs of Concentration Inequalities} \label{sec:concentration-proof}
We state deferred proofs for concentration inequalities in~\Cref{sec:random-regular-bipartite-Ising} here.

\oneSideWeightConcentration*
\begin{proof}
    When $\beta = 1$, $W(A) = n$ is deterministic.
    So we assume that $\beta \in [0, 1)$.
    We apply the following inequality: assume that there are $N$ elements with $K$ marked.
    Define the random variable $X$ as the number of marked elements in a picked uniform subset of size $M$.
    Then
    \begin{align} \label{eq:sampling-bound}
        \E{\beta^X} \le \e^{-\frac{(1 - \beta)KM}{N}}.
    \end{align}
    To show~\eqref{eq:sampling-bound}, for every element $i$, define $w_i = \beta$ if $i$ is marked and otherwise $w_i = 1$.
    Then by the Maclaurin's inequality,
    \[
        \E{\beta^X} = \frac{\sum_{1 \le i_1 < \ldots < i_M \le N} w_{i_1} \cdots w_{i_M}}{\binom{N}{M}} \le \ab(\frac{1}{N} \sum_{i = 1}^N w_i)^{M} \le \ab(1 - \frac{(1 - \beta)K}{N})^{M} \le \e^{-\frac{(1 - \beta)KM}{N}}.
    \]
    
    Now we turn to the lemma.
    Define $X$ be the $\Delta \abs{A}$ selected right half-edges and we uniformly give a random order to $X_1, \ldots, X_{\Delta\abs{A}}$.
    Define the Doob martingale
    \[
        \forall i = 1, \ldots, \Delta\abs{A}, \quad Z_i = \E{W(A) \mid X_1, \ldots, X_{i - 1}}.
    \]
    Now we fix a step $i = 1, \ldots, \Delta\abs{A}$ and a history $\calF_{i - 1} = \sigma(X_1, \ldots, X_{i - 1})$.
    Let $U$ be the remaining half-edges.
    For a possible next half-edge $y \in U$, define $f(y) \defeq \E{W(A) \mid \calF_{i - 1}, X_i = y}$.
    For two distinct possible half-edges $x, y \in U$, we bound $\abs{f(x) - f(y)}$ by construct a coupling of their remaining selected half-edges.
    When $X_i = x$, generate $F_x$ uniformly from the $(\abs{A} - i)$-size subset of $U \setminus \set{x}$.
    We construct $F_y$ from $F_x$ as
    \[
        F_y \defeq \begin{cases}
            F_x & y \notin F_x \\
            (F_x \setminus \set{y}) \cup \set{x} & y \in F_x
        \end{cases}\;.
    \]
    It is trivial to see $(F_x, F_y)$ is a valid coupling of the respective uniform distributions of possible remaining selected half-edges conditional on $X_i = x$ and $X_i = y$.
    When $y \in F_x$, it holds that $F_x \cup \set{x} = F_y \cup \set{y}$, meaning that the final selected subsets coincide.
    Therefore, they share the same value of $W(A)$.
    For $y \notin F_x$, by coupling, $F_x = F_y$ and thus the final selected subsets differ only at $x$ and $y$.
    In this case, the subset $F \defeq F_x = F_y$ is an uniform $(\Delta\abs{A} - i)$-size subset of $U \setminus \set{x, y}$.

    When $x, y$ are incident to a same vertex, the two cases show a same value of $W(A)$.
    Otherwise, suppose that $x$ is incident to $u$ and $y$ is incident to $v$.
    Let $z_u, z_v$ be the number of half-edges incident to $u$ and $v$ respectively hitting $A$ in the history profile $X_1, \ldots, X_{i - 1}$ and $Z_u, Z_v$ be the number of half-edges incident to $u$ and $v$ respectively hitting $A$ in the future profile $F$.
    The difference of two future values of $W(A)$ is at most
    \begin{align*}
        \abs{\ab(\beta^{z_u + 1 + Z_u} + \beta^{z_v + Z_v}) - \ab(\beta^{z_u + Z_u} + \beta^{z_v + 1 + Z_v})} &= (1 - \beta) \abs{\beta^{z_u + Z_u} - \beta^{z_v + Z_v}} \\
        &\le (1 - \beta)\ab(\beta^{z_u + Z_u} + \beta^{z_v + Z_v}).
    \end{align*}
    By~\eqref{eq:sampling-bound} with $N = \Delta n - i - 1$, $K = \Delta - z_u - 1$ and $M = \Delta\abs{A} - i$,
    \begin{align*}
        \E{\beta^{z_u + Z_u}} &= \beta^{z_u} \E{\beta^{Z_u}} \\
        &\le \beta^{z_u} \exp\ab(-(1 - \beta) \frac{(\Delta - z_u - 1)(\Delta\abs{A} - i)}{\Delta n - i - 1}) \\
        &\le \exp\ab(-(1 - \beta)\ab(z_u + \frac{\Delta - z_u - 1}{\Delta n - i - 1}\ab(\Delta \abs{A} - i))) \\
        &\le \exp\ab(-(1 - \beta)\frac{\Delta\abs{A} - i}{2n}).
    \end{align*}
    The same bound also holds for $v$.
    Therefore,
    \begin{align*}
        \abs{f(x) - f(y)} \le 2(1 - \beta)\e^{-(1 - \beta)(\Delta\abs{A} - i)/(2n)}.
    \end{align*}
    Then by Popoviciu's inequality on variances (\Cref{prop:Popoviciu-inequality}) and $1 - \e^{-x} \ge x/2$ on $[0, 1]$,
    \begin{align*}
        V_{\Delta\abs{A}} &\defeq \sum_{i = 1}^{\Delta\abs{A}} \E{(Z_i - Z_{i - 1})^2 \mid \calF_{i - 1}} \\
        &\le \frac{1}{4} \sum_{i = 0}^{\Delta\abs{A} - 1} 4(1 - \beta)^2\e^{-(1 - \beta)i/n} \\
        &\le \frac{(1 - \beta)^2}{1 - \e^{-(1 - \beta)/n}} \\
        &\le 2(1 - \beta)n.
    \end{align*}
    Trivially, $\abs{Z_i - Z_{i - 1}} \le 2(1 - \beta)$.
    Using~\Cref{lem:Freedman-inequality} with $v = 2(1 - \beta)n$ and $B = 2(1 - \beta)$, it holds that
    \begin{align*}
        \Pr{\abs{W(A) - \E{W(A)}} \ge t} &= \Pr{\abs{W(A) - \E{W(A)}} \ge t, V_{\Delta\abs{A}} \le 2(1 - \beta)n} \\
        &\le 2\exp\ab(-\frac{t^2}{2(1 - \beta)(2n + 2t/3)}).
    \end{align*}
    Thus we conclude the lemma.
\end{proof}

\degreeProfileConcentration*
\begin{proof}
    For a subset $A \subseteq V_L$ of size $\theta n$, the probability of an admissible degree profile $p$ is
    \[
        P \defeq \frac{\binom{n}{np(0), \ldots, np(\Delta)} \prod_{j = 0}^{\Delta} \binom{\Delta}{j}^{n p(j)}}{\binom{\Delta n}{\Delta \theta n}}.
    \]
    We use the following crude bound
    \[
        \binom{n}{np_0, \ldots, np_{\Delta}} \le \e^{nH(p)}, \quad \binom{\Delta n}{\Delta \theta n} \ge \frac{\e^{\Delta n H(\theta)}}{\Delta n + 1}
    \]
    and obtain that
    \[
        P \le (\Delta n + 1) \e^{n H(p) - \Delta n H(\theta)} \prod_{j = 0}^{\Delta} \binom{\Delta}{j}^{np(j)}.
    \]
    On the other side, for every admissible $p$, it holds that $\sum_{j} p(j) = 1$ and $\sum_{j} j p(j) = \Delta \theta$.
    Therefore,
    \begin{align*}
        \KL{p}{\Phi_{\theta}} &= \sum_{j = 0}^{\Delta} p(j) \ln{\frac{p(j)}{\binom{\Delta}{j} \theta^j (1 - \theta)^{\Delta - j}}} \\
        &= \Delta H(\theta) - H(p) - \sum_{j = 0}^{\Delta} p(j) \ln{\binom{\Delta}{j}}.
    \end{align*}
    Substituting the identity into the upper bound, it holds that
    \begin{align*}
        \Pr{p_A = p} \le (\Delta n + 1) \e^{-n\KL{p}{\Phi_{\theta}}}.
    \end{align*}
    By union bound and the definition of $r_n(u)$,
    \begin{align*}
        \Pr{\exists A~\text{violates~\eqref{eq:degree-profile-concentration}}} &\le 2 \sum_{k = 0}^n \binom{n}{k} (n + 1)^{\Delta + 1} (\Delta n + 1) \exp\ab(-n\ab(H(k/n) + r_n(u))) \\
        &\le 2n(\Delta n + 1) (n + 1)^{\Delta + 1} \exp\ab(-n r_n(u)) \\
        &\le \e^{-u}.
    \end{align*}
    Thus we conclude the lemma.
\end{proof}

\end{document}